\def\year{2021}\relax
\documentclass[letterpaper]{article} %
\usepackage{aaai21}  %
\usepackage{times}  %
\usepackage{helvet} %
\usepackage{courier}  %
\usepackage[hyphens]{url}  %
\usepackage{graphicx} %
\usepackage{natbib}  %
\usepackage{caption} %
\usepackage{amsmath}

\usepackage{algorithm}
\usepackage[noend]{algorithmic}

\usepackage{amssymb}
\usepackage{amsmath}
\usepackage{mathtools}
\usepackage{amsthm}
\usepackage{aliascnt}
\usepackage{physics}
\usepackage{accents}
\usepackage[shortlabels]{enumitem}
\usepackage{makecell}
\usepackage{multirow}
\usepackage{etoolbox}

\usepackage{color}
\definecolor{green}{rgb}{0,0.5,0} %

\usepackage{cleveref}
\usepackage{autonum} %
\makeatletter
\autonum@generatePatchedReferenceCSL{Cref}
\autonum@generatePatchedReferenceCSL{cref}
\makeatother

\newtheorem{theorem}[equation]{Theorem}

\newtheorem{lemma}[equation]{Lemma}
\newtheorem{proposition}[equation]{Proposition}
\newtheorem{corollary}[equation]{Corollary}

\newtheorem{warning}[equation]{Warning}
\theoremstyle{definition}
\newtheorem{definition}[equation]{Definition}

\numberwithin{equation}{section}
\makeatletter
\let\c@algorithm\relax
\makeatother
\newaliascnt{algorithm}{equation}

\DeclareMathOperator*{\E}{\mathbb E}

\newcommand{\R}{\mathbb R}

\newcommand{\ubar}[1]{\underaccent{\bar}{#1}}

\let\op\operatorname
\let\eps\varepsilon
\let\mc\mathcal

\setlist[enumerate,1]{label={(\arabic*)}}

\title{Finding and Certifying (Near-)Optimal Strategies \\ in Black-Box Extensive-Form Games}

\author{

    Brian Hu Zhang,\textsuperscript{\rm 1}
    Tuomas Sandholm\textsuperscript{\rm 1,2,3,4}
    \\
}
\affiliations{

    \textsuperscript{\rm 1}Computer Science Department, Carnegie Mellon University\\
    \textsuperscript{\rm 2}Strategic Machine, Inc.\\
    \textsuperscript{\rm 3}Strategy Robot, Inc.\\
    \textsuperscript{\rm 4}Optimized Markets, Inc. \\
    \{bhzhang, sandholm\}@cs.cmu.edu \\

}

\begin{document}

\maketitle

\begin{abstract}
Often---for example in war games, strategy video games, and financial simulations---the game is given to us only as a black-box simulator in which we can play it. In these settings, since the game may have unknown nature action distributions (from which we can only obtain samples) and/or be too large to expand fully, it can be difficult to compute strategies with guarantees on exploitability. Recent work \cite{Zhang20:Small} resulted in a notion of certificate for extensive-form games that allows exploitability guarantees while not expanding the full game tree. However, that work assumed that the black box could sample or expand arbitrary nodes of the game tree at any time, and that a series of exact game solves (via, for example, linear programming) can be conducted to compute the certificate. Each of those two assumptions severely restricts the practical applicability of that method. In this work, we relax both of the assumptions. We show that high-probability certificates can be obtained with a black box that can do nothing more than play through games, using only a regret minimizer as a subroutine. As a bonus, we obtain an equilibrium-finding algorithm with $\tilde O(1/\sqrt{T})$ convergence rate in the extensive-form game setting that does not rely on a sampling strategy with lower-bounded reach probabilities (which MCCFR assumes). We demonstrate experimentally that, in the black-box setting, our methods are able to provide nontrivial exploitability guarantees while expanding only a small fraction of the game tree.
\end{abstract}

\section{Introduction}
Computational equilibrium finding has led to many recent breakthroughs in AI in games such as poker \cite{Bowling15:Heads,Brown17:Superhuman,Moravvcik17:DeepStack,Brown19:Superhuman} where the game is fully known. However, in many applications, the game is not fully known; instead, it is given only via a simulator that permits an algorithm to play through the game repeatedly (e.g., \citealt{Wellman06:Methods,Lanctot17:Unified,Tuyls18:Generalised,Areyan20:Improved}). 
The algorithm may never know the game exactly. While deep reinforcement learning has yielded strong practical results in this setting~\citep{Vinyals19:Grandmaster,Berner19:Dota}, those methods lack the low-exploitability guarantees of game-theoretic techniques, even with infinite samples and computation. Furthermore, the standard method of evaluating exploitability of a strategy---computing the equilibrium gap of the strategy---is to compute a best response for each player. This, however, assumes the whole game to be known exactly.

Recently, \citet{Zhang20:Small} defined a notion of {\em certificate} for imperfect-information extensive-form games that can address these problems.\footnote{\citet{Gatti11:Equilibrium} explore the problem in extensive-form games of {\em perfect} information with infinite strategy spaces.} A certificate enables verification of the exploitability of a given strategy without exploring the whole game tree. However, that work has a few limitations that reduce its practical applicability. First, they assume a black-box model that allows sampling or expanding arbitrary nodes in the game tree. Yet most simulators only allow the players to start from the root of the game, and chance nodes in the simulator affect the path of play, so exploration by jumping around in the game tree is not supported. Second, their algorithm requires an exact game solver, for example, a \textit{linear program (LP)} solver, to be invoked repeatedly as a subroutine. This reduces the ability of the algorithm to scale to cases in which LP is impractical due to run time or memory considerations. 

In this paper, we address both of these concerns. We give algorithms that create certificates in extensive-form games in a simple black-box model, with either an exact game solver or a regret minimizer as a subroutine. We show that our algorithms achieve convergence rate $O(\sqrt{\log(T)/T})$ (hiding game-dependent constants). This matches, up to a logarithmic factor, the convergence rate of regret minimizers such as \textit{counterfactual regret minimization (CFR)}~\cite{Zinkevich07:Regret,Brown19:Solving} or its stochastic variant, \textit{Monte Carlo CFR (MCCFR)}~\cite{Lanctot09:Monte,Farina20:Stochastic}---while also providing verifiable equilibrium gap guarantees unlike those prior techniques suitable for the black-box setting. In particular, we are also able to compute ex-post exploitability bounds without knowing the exact game or iterating through all its sequences. In contrast, previous techniques would need to rely on either their worst-case convergence bounds, which are at least linear (and usually worse) in the number of information sets~\cite{Lanctot09:Monte}, or else know the exact game to perform an exact best-response computation. We prove that this convergence rate is optimal for the setting. We demonstrate experimentally that our method allows us to construct nontrivial certificates in games with good sample efficiency, namely, while taking fewer samples than there are nodes in the game.

As a side effect, our algorithm is, to our knowledge, the first extensive-form game-solving algorithm that enjoys $\tilde O(N/\sqrt{T})$ (where $N$ is the number of nodes) Nash gap convergence rate in two-player zero-sum games in the model-free setting, without the problematic assumption of having an \textit{a priori} strategy with known lower-bounded reach probabilities on all nodes that is required by MCCFR.  \citet{Farina21:Bandit} only achieves $\op{poly}(N)/\sqrt{T}$ convergence against an {\em arbitrary, fixed} strategy (which is not an exploitability bound), and \citet{Farina21:Model} has weaker convergence rate $\op{poly}(N)/T^{1/4}$. 
Unlike in the latter two papers, which are regret minimizers, though, we are controlling both players during the learning.

Our techniques also work for games where payoffs can be received at internal nodes (not just at leaves), and for coarse-correlated equilibrium in general-sum multi-player games.

{\em Game abstraction} has commonly been used to reduce game tree size prior to solving~\cite{Billings03:Approximating,Gilpin06:Competitive,Brown15:Simultaneous,Cermak17:An}. Practical abstraction techniques without exploitability guarantees were used in achieving superhuman performance in no-limit Texas hold'em poker in the {\em Libratus}~\citep{Brown17:Superhuman} and {\em Pluribus}~\citep{Brown19:Superhuman} agents. 
There has been research on abstraction algorithms with exploitability guarantees for specific settings~\cite{Sandholm12:Lossy,Basilico11:Automated} and for general extensive-form games (e.g.,~\citealt{Gilpin07:Lossless,Lanctot12:No,Kroer14:Extensive,Kroer15:Discretization,Kroer16:Imperfect,Kroer18:Unified}), but these are not scalable for large games such as no-limit Texas hold'em, and the guarantees depend on the difference between the abstracted game and the real game being {\em known}.

\section{Notation and Background}

We study {\em extensive-form games}, hereafter simply {\em games}. An extensive-form game consists of:
\begin{enumerate}
\item a set of players $\mc P$, usually identified with positive integers $1, 2, \dots n$. {\em Nature}, a.k.a. {\em chance}, will be referred to as player 0. For a given player $i$, we will often use $-i$ to denote all players except $i$ and nature.
\item a finite tree $H$ of {\em nodes}, rooted at some {\em root node} $\emptyset$. The edges connecting a node $h$ to its children are labeled with {\em actions}. The set of actions at $h$ will be denoted $A(h)$. $h \preceq z$ means $z$ is a descendant of $h$, or $z = h$.
\item a map $P : H \to \mc P \cup \{0\}$, where $P(h)$ is the player who acts at node $h$ (possibly nature).
\item for each player $i$, a {\em utility function} $u_i : H \to \R$. It will be useful for us to allow players to gain utility at {\em internal} nodes of the game tree. Along any path $(h_1, h_2, \dots, h_k)$, define $u(h_1 \to h_k) = \sum_{i=1}^k u(h_i)$ to be the total utility gained along that path, including both endpoints. The goal of each player is to maximize their total reward $u(\emptyset \to z)$, where $z$ is the terminal node that is reached.
\item for each player $i$, a partition of player $i$'s decision points, i.e., $P^{-1}(i)$, into \textit{information sets} (or {\em infosets}). In each infoset $I$, every $h \in I$ must have the same set of actions.
\item for each node $h$ at which nature acts, a distribution $\sigma_0(\cdot |h)$ over the actions available to nature at node $h$.
\end{enumerate}

We will use $(G, u)$, or simply $G$ when the utility function is clear, to denote a game. $G$ contains the tree and information set structure, and $u = (u_1, \dots, u_n)$ is the profile of utility functions. 

For any history $h \in H$ and any player $i \in \mc P$, the {\em sequence} $s_i(h)$ of player $i$ at node $h$ is the sequence of information sets observed and actions taken by $i$ on the path from $\emptyset$ to $h$. In this paper, all games will be assumed to have {\em perfect recall}: if $h_1, h_2 \in I$ and $i$ acts at $I$, then $s_i(h_1) = s_i(h_2)$. 

A {\em behavior strategy} (hereafter simply {\em strategy}) $\sigma_i$ for player $i$ is, for each information set $I$ at which player $i$ acts, a distribution $\sigma_i(\cdot | I)$ over the actions available at that infoset. When an agent reaches information set $I$, it chooses action $a$ with probability $\sigma_i(a | I)$. A tuple $\sigma = (\sigma_1, \dots, \sigma_n)$ of behavior strategies, one for each player $i \in \mc P$, is a {\em strategy profile}. A distribution over strategy profiles is called a {\em correlated strategy profile}, and will also be denoted $\sigma$. The {\em reach probability} $\sigma_i(h)$ is the probability that node $h$ will be reached, assuming that player $i$ plays according to strategy $\sigma_i$, and all other players (including nature) always choose actions leading to $h$ when possible. Analogously, we define $\sigma(h) = \prod_{i \in \mc P \cup \{0\}} \sigma_i(h)$ to be the probability that $h$ is reached under strategy profile $\sigma$. This definition naturally extends to sets of nodes or to sequences by summing the reach probabilities of all relevant nodes. %

Let $S_i$ be the set of sequences for player $i$. The {\em sequence form} of a strategy $\sigma_i$ is the vector $x \in \R^{S_i}$ given by $x[s] = \sigma_i(s)$. The set of all sequence-form strategies is the {\em sequence form strategy space} for $i$, and is a convex polytope~\cite{Koller94:Fast}.

The {\em value} of a profile $\sigma$ for player $i$ is $u_i(\sigma):= \E_{z \sim \sigma} u_i(\emptyset \to z)$. The future utility of a profile starting at $h$ is $u(\sigma|h)$; that is, 
$
u(\sigma|h) = \E_{z \sim \sigma|h} u(h \to z).
$ 

The {\em best response value} $u^*_i(\sigma_{-i})$ for player $i$ against an opponent strategy $\sigma_{-i}$ is the largest achievable value; i.e., $u^*_i(\sigma_{-i}) = \max_{\sigma_i} u_i(\sigma_i, \sigma_{-i})$. A strategy $\sigma_i$ is an $\eps$-{\em  best response} to opponent strategy $\sigma_{-i}$ if $u_i(\sigma_i, \sigma_{-i}) \ge u^*_i(\sigma_{-i}) - \eps$. A {\em best response} is a $0$-best response.

A strategy profile $\sigma$ is an {\em $\eps$-Nash equilibrium} (which we will call {\em $\eps$-equilibrium} for short) if all players are playing $\eps$-best responses. A {\em Nash equilibrium} is a $0$-Nash equilibrium.

We also study finding certifiably good strategies for the game-theoretic solution concept called \textit{coarse-correlated equilibrium}. In such equilibrium, if $\sigma$ is correlated, the deviations $\sigma_i$ when computing best response are \textit{not} allowed to depend on the shared randomness. A correlated strategy profile $\sigma$ is a {\em coarse-correlated $\eps$-equilibrium} if all players are playing $\eps$-best responses under this restriction.

\subsection{$\eps$-Equilibria within Pseudogames}
We now define {\em pseudogames}, first introduced by \citet{Zhang20:Small}.
\begin{definition}
	A {\em pseudogame} $(\tilde G, \alpha, \beta)$ is a game in which some nodes do not have specified utility but rather have only lower and upper bounds on utilities. Formally, for each player $i$, instead of the standard utility function $u_i$, there are lower and upper bound functions $\alpha_i, \beta_i : H \to \R$.
\end{definition}
We will always use $\Delta$ to mean $\beta - \alpha$. For $\alpha$ and $\beta$, we will use the same notation overloading as we do for the utility function $u$. For example, $\alpha(h \to z)$ and $\alpha(\sigma|h)$ have the corresponding meanings.

\begin{definition}\label{def:trunk}
	$(\tilde G, \alpha, \beta)$ is a {\em trunk} of a game $(G, u)$ if: %
	\begin{enumerate}
		\item $\tilde G$ can be created by collapsing some internal nodes of $G$ into terminal nodes (and removing them from information sets they are contained in), and
		\item for all nodes $h$ of $G$, all players $i$, and all strategy profiles $\sigma$, we have $\alpha_i(\sigma|h) \le u_i(\sigma|h) \le \beta_i(\sigma|h)$.
	\end{enumerate}  
\end{definition}
It is possible for information sets to be partially or totally removed in a trunk game. 
Next we state the basics of equilibrium and coarse-correlated equilibrium in pseudogames.
\begin{definition}\label{def:pseudonash}
	A {\em (coarse-correlated) $\eps$-equilibrium} of 
	$(\tilde G, \alpha, \beta)$ is a (correlated) profile $\sigma$ such that  the {\em equilibrium gap} $ \beta_i^*(\sigma_{-i}) - \alpha_i(\sigma) $ of each player $i$ is at most $\eps$.
\end{definition}
\begin{definition}
	A {\em (coarse-correlated) $\eps$-certificate} for a game $G$ is a pair $(\tilde G, \sigma)$, where $\tilde G$ is a trunk of $G$ and $\sigma$ is a (coarse-correlated) $\eps$-equilibrium of $\tilde G$.
\end{definition}
\begin{proposition}[\citealt{Zhang20:Small}]\label{prop:reasonable}
	Let $(\tilde G, \sigma)$ be an $\eps$-certificate for game $G$. Then any strategy profile in $G$ created by playing according to $\sigma$ in any information set appearing in $\tilde G$ and arbitrarily at information sets not appearing in $\tilde G$ is a $\eps$-equilibrium in $G$.
\end{proposition}
Though the above proposition was stated only for Nash equilibrium by~\citet{Zhang20:Small}, we observe that it applies to coarse-correlated equilibria as well.

\subsection{The Zero-Sum Case}
A two-player game is {\em zero sum} if $u_1 = -u_2$. In this case, we refer to a single utility function $u$; it is understood that Player~2's utility function is $-u$. In zero-sum games, all equilibria have the same expected value; this is called the {\em value of the game}, and we denote it by $u^*$.
In the zero-sum case, we use a slightly different notion of $\eps$-equilibrium of a pseudogame, which will make the subsequent results tighter.
\begin{definition}
	A two-player pseudogame $(\tilde G, \alpha, \beta)$ is zero-sum if $\alpha_2 = -\beta_1$ and $\beta_2 = -\alpha_1$.
\end{definition}
As alluded to above, in this situation, we will drop the subscripts, and write $\alpha$ and $\beta$ to mean $\alpha_1$ and $\beta_1$. In particular, $(\tilde G, \alpha)$ and $(\tilde G, \beta)$ are zero-sum games. 
\begin{definition}\label{def:zerosumnash}
	An $\eps$-equilibrium of a two-player zero-sum pseudogame $(\tilde G, \alpha, \beta)$ is a profile $(x^*, y^*)$ for which the {\em Nash gap} $
	\beta^*(y^*) - \alpha^*(x^*)$ is at most $\eps$.
\end{definition}
In zero-sum games, we need not concern ourselves with correlation, since at least one player can always deviate to playing independently of the other player and not lose utility. In particular, a coarse-correlated $\eps$-equilibrium remains an $\eps$-equilibrium if the correlations are removed.

\subsection{Regret Minimizers}
{\em Online convex optimization} (OCO) \cite{Zinkevich03:Online} is a rich framework through which to understand decision-making in possibly adversarial environments. 
\begin{definition}
Let $X \subseteq \R^n$ be a compact, convex set. A {\em regret minimizer} $\mc A_X$ on $X$ is an algorithm that acts as follows. At each time $t = 1, 2, \dots T$, the algorithm $\mc A_X$ outputs a {\em decision} $x^t \in X$, and then receives a linear {\em loss} $\ell_t : X \to \R$, which may be generated adversarially.
\end{definition}
The goal is to minimize the {\em cumulative regret} 
\begin{align}
    R_T := \max_{x \in X} \sum_{t=1}^T \qty[\ell_t(x^t) - \ell_t(x)].
\end{align}
For example, CFR and its modern variants achieve $O(\sqrt{T})$ regret in sequence-form strategy spaces. 

The connection between OCO and equilibrium-finding in games is via the following observation. Let $(\sigma^1, \dots, \sigma^T)$ be any sequence of strategy profiles, and let $\bar \sigma$ be the correlated strategy profile that is uniform over $\sigma^1, \dots, \sigma^T$. Suppose that player $i$ generated her strategy at each timestep via a regret minimizer, and achieved regret $R_T$. Then, by definition of regret, $i$ is playing an $\eps$-best response to $\bar \sigma$, where $\eps = R_T/T$. Thus, in particular, if  all players are playing using a regret minimizer with sublinear regret, the average strategy profile $\bar \sigma$ converges to a coarse-correlated equilibrium.
\section{Black-Box Model and Problem Statement}\label{s:statement}
Let $(G, u)$ be an $n$-player game, which we assume to be given to us as a black box. Given a profile $\sigma$, the black box allows us to sample a playthrough from $G$ under $\sigma$. We also assume that, at every node $h$, the black box gives the actions and corresponding child nodes available at $h$, as well as correct (but not necessarily tight) bounds $[\alpha(h \to *), \beta(h \to *)]$ on the utility $u(h \to z)$ of every terminal descendant $z \succeq h$:
{\footnotesize
\begin{align}
    \alpha(h \to *) \le \min_{z \succeq h} u(h \to z) \le \max_{z \succeq h} u(h \to z) \le \beta(h \to *).
\end{align}
}%
Our goal in this paper is to develop equilibrium-finding algorithms that give {\em anytime, high-probability, instance-specific} exploitability guarantees that can be computed without expanding the rest of the game tree, and are better than the generic guarantees given by the worst-case runtime bounds of algorithms like MCCFR. More formally, our goal is, after $t$ playthroughs, to efficiently maintain a strategy profile $\sigma^t$ and bounds $\eps_{i,t}$ on the equilibrium gap of each player's strategy (or, in the zero-sum case, a single bound $\eps_t$ on the Nash gap) that are correct with probability $1 - 1/\op{poly}(t)$. 

\section{Lower Bound}
Before proceeding to algorithms, we prove a lower bound on the sample complexity of computing such a strategy profile. Let $\gamma > 0$ be an arbitrary constant. Consider a multi-armed bandit instance in which the left arm has some unknown reward distribution over $\{0, 1\}$, and the right arm always gives utility $1/2$. Let $p$ be the probability that the left arm gives $1$. We will consider the two games, $G_-$ and $G_+$, in which, respectively, the left arm gives $p = 1/2 - \eps$ and $p = 1/2 + \eps$, where $\eps = \Theta(\sqrt{\gamma \log(t)/t})$, and the $\Theta$ hides only an absolute constant. Suppose $t$ samples of the left arm are taken (the right arm would not be sampled by an optimal algorithm, since its value is already known). We will say that the algorithm has {\em selected the correct arm} if $\sigma^t$ assigns a higher probability to the better arm than it does to the worse arm. Then the following two facts are simultaneously true.
\begin{enumerate}
\item By binomial tail bounds, no algorithm can select the  correct arm with probability better than $1 - \Theta(1/t^\gamma)$.
\item In the event that an algorithm fails to select the correct arm at time $t$, its equilibrium gap is $\Theta(\eps)$. 
\end{enumerate}
Thus, we have the following theorem.\footnote{The dependence of the game on $t$ in the above argument is fine. To prove the theorem, we only need to give a game $G$ and time $t$ for which no algorithm can achieve $\eps_{i, t} = O(\sqrt{\log(t)/t})$ with sufficiently high probability.}
\begin{theorem}
    Any algorithm that provides the guarantees described in 
    \Cref{s:statement} must have $\eps_{i, t} = \Omega(\sqrt{\log(t)/t})$. 
\end{theorem}
We will now describe algorithms matching this bound. 

\section{Exploration and Confidence Sequences}
We now describe our main theoretical construction: a notion of \textit{confidence sequence} for games, that enables us to construct high-probability certificates from playthroughs. Let $\mc A$ be an exploration policy that generates, possibly adaptively, a profile $\sigma^t$ at each timestep $t$.
\begin{definition}\label{def:confidence-sequence}
A {\em confidence sequence} for a game $G$ is a sequence of pseudogames $(\hat G^t, \hat \alpha^t, \hat \beta^t)$ created by the following protocol. Start with $\hat G^0$ containing only a root node $\emptyset$ and trivial reward bounds (that is, $\hat \beta^0(\emptyset) = \beta(\emptyset \to *) $ and $\hat \alpha^0(\emptyset) = \alpha(\emptyset \to *) $). At each time $t$:
\begin{enumerate}
\item Query $\mc A$ to obtain profile $\sigma^t$
\item Play a single game of $G$ according to $\sigma^t$.
\item Create $\hat G^{t}$ from $\hat G^{t-1}$ as follows.
\begin{enumerate}
\item Expand all nodes\footnote{It is also valid to expand only the first new node on the path of play, that is, the first node on the sampled trajectory that is not previously expanded. That does not change any of our theoretical results. To {\em expand} a node means to add all of its children to the game tree, adjusting the upper- and lower-bound functions as necessary. Other nodes at a given information set are {\em not} necessarily added when one node is; it is therefore possible for information sets to be partially added in a pseudogame. } on the path of play.
\item For each chance node $h$ in $\hat G^t$:
\begin{enumerate}[(i)]
\item If $h$ was encountered on the path of play, update $\hat\sigma^{t}_0(\cdot|h)$ for each according to the action observed at $h$ to be the empirical frequency of play.
\item Let 
\begin{align}\label{eq:rho}
\rho(h) = \sqrt{\frac{1}{2 t_h} \qty(\abs{A_h} \log 2  +  \log t^2 C_tn)}.
\end{align}
where $t_h$ is the number of times $h$ has been sampled (including on this iteration), and $C_t$ is the number of chance nodes in $\hat G^{t}$. Now set $\hat\beta^t_i(h) = u_i(h) + \rho(h) \Delta(h \to *)$, and $\hat\alpha^t_i(h) = u_i(h) - \rho(h) \Delta(h \to *)$.
\end{enumerate}
\end{enumerate}
\end{enumerate}
\end{definition}

We will use $(G^t, \alpha^t, \beta^t)$ to denote the pseudogame with the same game tree as $\hat G^t$, but with the exact correct nature probabilities (that is, no sampling error, and $\rho(h) = 0$). 

\begin{theorem}[Correctness]\label{thm:ucb}
	For any time $t$ and exploration policy $\mc A$, with probability at least $1 - 2/t^2$, for every profile $\sigma$ and player $i$, we have $\hat \alpha_i^t(\sigma) \le \alpha_i(\sigma) \le \beta_i(\sigma) \le \hat \beta_i^t(\sigma)$.
\end{theorem}
Proofs are in the appendix, which can be found on the arXiv version at https://arxiv.org/abs/2009.07384.
In this case, we will call the sequence {\em correct at time $t$}.
These probabilities can be strengthened to any inverse polynomial function of $t$ by replacing $t^2$ in \Cref{eq:rho} with a suitably larger polynomial.

Extra domain-specific information about the chance distributions can easily be incorporated into the bounds. For example, if two chance nodes are known to have the same action distribution, their samples can be merged. If the distribution of a chance node is known exactly, no sampling is necessary at all, and the number of chance nodes $C_t$ in \Cref{eq:rho} may be decremented accordingly.

\begin{definition}
For an exploration policy $\mc A$ creating a confidence sequence $(\hat G^t, \hat \alpha^t, \hat \beta^t)$, the {\em cumulative uncertainty} $U_{i,T}$ for player $i$ after the first $T$ iterations is given by 
\begin{align}
U_{i,T} := \sum_{t=1}^T \hat \Delta^t_i(\sigma^t).
\end{align}
\end{definition}
This can be thought of as the regret of an online optimizer that plays $\sigma^t$ at time $t$, and then observes loss $\hat \beta^t_i - \hat \alpha^t_i$. In a sense, the next result is the main theorem of our paper, and we find it the most surprising result of the paper. All our convergence guarantees stated later in the paper rely on it. 
\begin{theorem}\label{thm:main}
	Suppose that the true rewards are bounded in $[0, 1]$. Then for all times $T$, all players $i$, and {\em any} exploration policy $\mc A$, we have
	\begin{align}
	\E U_{i,T} \le 2C_{T} \sqrt{2TM} + N_T
	\end{align}
	where $N_T$ is the number of total nodes in $\hat G_T$, 
	\begin{align}
	M = \max_{\textup{chance nodes $h$}} \qty( \abs{A_h} \log 2+  \log {2T^2 C_Tn}),
	\end{align}
	and the expectation is over the sampling of games and (if applicable) the randomness of $\mc A$.
\end{theorem}
$M$ is a constant that depends on the final pseudogame $\hat G^T$. Importantly, it does not depend on the actual game $G$! This makes it possible for our approach to give meaningful exploitability guarantees while not exploring the full game. For fixed underlying game and confidence, $M$ increases as $\Theta(\log T)$, and hence $U_{i,T}$ increases as $O(\sqrt{T \log T})$ for a fixed game.

\section{Solving Games via Confidence Intervals}
The above discussion leads naturally to algorithms that generate certificates, which we will now discuss.
\begin{algorithm}[H]
	\caption{Two-player zero-sum certificate finding}\label{alg:zerosum}
	\begin{algorithmic}[1]
		\STATE {\bf Input:} black-box zero-sum game
		\STATE Initialize confidence sequence $(\hat G^0, \hat \alpha^0, \hat \beta^0)$ 
		\FOR{$t = 1, 2, \dots$}
		\STATE Solve $(\hat G^{t-1}, \hat \alpha^{t-1})$ and $(\hat G^{t-1}, \hat \beta^{t-1})$ exactly to obtain equilibria $(\ubar x^{t-1}, \ubar y^{t-1})$ and $(\bar x^{t-1}, \bar y^{t-1})$.\label{step:gamesolve}
		\STATE Create next pseudogame $\hat G^{t}$ by sampling one playthrough according to some profile $\sigma^{t}$\label{step:sample}
		\ENDFOR
	\end{algorithmic}
\end{algorithm}
\begin{definition}
The {\em Nash gap bound} $\eps_t$ at time $t$ of \Cref{alg:zerosum} is the difference between the values of the games with utility functions $\hat \beta^{*t}$ and $\hat \alpha^{*t}$. Formally, $\eps_t = \hat \beta^{*t} - \hat \alpha^{*t}$.
\end{definition}
\begin{proposition}\label{prop:cert-lp-eqm}
	Assuming that the confidence sequence is correct at time $t$, the pessimistic equilibrium $(\ubar x^t, \bar y^t)$ computed by \Cref{alg:zerosum} is an $\eps_t$-equilibrium of $G^t$.
\end{proposition}
This allows us to {\em know} (with high probability) when we have found an $\eps$-equilibrium, without expanding the remainder of the game tree, even in the case when chance's strategy is not directly observable. 
The choice of exploration policy in Line~\ref{step:sample} is very important. We now discuss that.
\begin{definition}
The {\em optimistic exploration policy} is $\sigma^t = (\bar x^{t-1}, \ubar y^{t-1})$; that is, both players explore according to the optimistic pseudogame.
\end{definition}
\begin{proposition}\label{prop:uncertainty-nash-gap}
Under the optimistic policy, $\eps_t \le \hat\Delta^t(\sigma^t)$.
\end{proposition}

Together with \Cref{thm:main}, this immediately gives us a convergence bound on \Cref{alg:zerosum}:
\begin{corollary}\label{thm:main-lp}
Suppose we use optimistic exploration, and the true game $G$ has rewards bounded in $[0,1]$. Let $\eps^*_T$ be the known bound on the Nash gap of the best pessimistic equilibrium found so far; that is, $\eps^*_T = \min_{t \le T} \eps_t$. Then
\begin{align}
\E\eps_T^* \le 2 C_T \sqrt{\frac{2M}{T}} + \frac{1}{T} N_T.
\end{align}
where again the expectation is over the sampling of games and randomness of $\mc A$.
\end{corollary}
We thus use optimistic exploration in our experiments (\Cref{s:experiments}). 
This is {\em not} the same kind of bound that is achieved by MCCFR and related algorithms. Those algorithms guarantee an upper bound on exploitability as a function of {\em total number of iterations}; here, we bound the number of {\em samples}. After every sample, our \Cref{alg:zerosum} solves the entire pseudogame generated so far. This may be expensive (though, since the game solves can be implemented as LP solves with warm starts from the previous iteration, in practice they are still reasonably efficient). However, as in \citet{Zhang20:Small}, our convergence guarantee has the distinct advantage of being dependent only on the current pseudogame, not the underlying full game. Furthermore, as the experiments later in this paper show, in practice, $\eps^*_T$ is usually significantly smaller than its worst-case bound.

In several special cases, \Cref{alg:zerosum} corresponds naturally to known algorithms and results. 
\begin{enumerate}
\item {\em Perfect information and deterministic}: Assuming the game solves return pure strategies (which is always possible here), \Cref{alg:zerosum} is exactly the same as Algorithm~6.7 of \citet{Zhang20:Small}. In particular, in the two-player case, it is equivalent to incremental alpha-beta search; in the one-player case, it is equivalent to A* search~\cite{Hart68:Formal}, where the upper bound $\beta(h \to *)$ corresponds to the heuristic lower bound on the total distance from the root to the goal.
\item {\em Nature probabilities known}: \Cref{alg:zerosum} is very similar (but not identical, due to the simpler black-box model) to Algorithm~6.7 of \citet{Zhang20:Small}.
\item {\em Multi-armed stochastic bandit}: \Cref{alg:zerosum} is, up to a constant factor in \Cref{eq:rho}, equivalent to UCB1~ \cite{Auer02:Finite}, and \Cref{thm:main-lp} matches the worst-case $O(\sqrt{T \log T})$ dependence on $T$ in the regret bound of UCB1. The worse dependence on the number of arms can be remedied by a more detailed analysis, which we skip here.
\end{enumerate}

 \Cref{alg:zerosum} and Algorithm~6.7 from~\citet{Zhang20:Small} may seem similar to the extensive-form double-oracle algorithm~\cite{Bosansky14:Exact}. However, ours does not compute best responses in the full game $G$, but rather only in the current pseudogame $\tilde G$. Thus, we maintain the guarantee of never needing to use any information about the game beyond what is given by the simulator during the limited sampling---neither while solving nor while computing exploitability bounds. To our knowledge, this work and~\citet{Zhang20:Small} are the first to achieve this guarantee.

In practice, due to the computational cost of the game solves, we recommend running several samples per game solve. This enhances computational efficiency in domains where the game is not prohibitively large for LPs, or samples are relatively fast to obtain. 

\section{Faster Iterates via Regret Minimization}
A major weakness of \Cref{alg:zerosum} is its reliance on an exact game solver as a subroutine, which can be slow or even infeasible computationally. Could we replace the exact solver with a single iteration of some iterative game solver, and still maintain the $\tilde O(1/\sqrt{T})$ convergence rate? In this section we show how to do this with regret minimizers.

\subsection{Extendable Regret Minimizers}

We now define a class of regret minimizers, which we coin \textit{extendable}, which we can use to achieve the goal mentioned above. Intuitively, for an extendable family of regret minimizers, expending a leaf of the pseudogame 
does not change the behavior or regret of the regret minimizer, so long as the past losses do not depend on the actions taken at the new information set, which is always the case with our algorithms because they have never visited the new information set. Thus, when working with a extendable family $\mc A$, it makes sense to speak about ``running $\mc A$ on a game $G$'', even if information sets may be added to $G$ over time. We will exploit this language.
For example, CFR (thus also MCCFR, since it is nothing but CFR with stochastic gradient estimates \cite{Farina20:Stochastic}) is a extendable family. In this case, the function $\phi$ described below simply initializes regrets at the new information set to $0$.

Formally, let $\mc L(X)$ be the set of linear functions on $X$.  Consider a regret minimizer $\mc A_X$ on $X$. We will think of $\mc A_X$ as maintaining a {\em state} $s_t \in S_X$. At any time $t$, the algorithm outputs strategy $x_t \gets x_X(s_t)$ for some map $x_X : S_X \to X$, and after observing loss $\ell_t$, the algorithm updates the state via $s_{t+1} \gets u_X(s_t, \ell_t)$, where $u_X : S_X \times \mc L(X) \to S_X$ is an update function. As such, $\mc A_X$ can be thought of as a pair $(x_X, u_X)$. 
For example, when $X$ is the $n$-simplex and $\mc A_X$ is regret matching~\cite{Hart00:Simple}, $S_X$ is $\R^n$, the update function is $u_X(s_t, \ell_t) = s_t + \ell_t - \ev{\ell_t, x_X(s_t)}$, and the strategy is $x_C(s_t)(i) \propto [s_t(i)]_+$.
\begin{definition}
Let $\mc A = \{ \mc A_X\}$ be a family of regret minimizers, one for each extensive-form strategy space $X$. $\mc A$ is {\em extendable}
if for every $X$ and every $X' \subseteq X \times \R^m$ formed by adding a decision point (with $m$ actions) to $X$, there is a function $\phi : S_X \to S_{X'}$ such that for every state $s \in S_X$:
\begin{enumerate}
\item $x_{X'}(\phi(s))$ agrees with $x_X(s)$ in $X$, and
\item for every loss function $\ell \in \mc L(X)$, we have
$\phi(u_X(s, \ell)) = u_{X'}(\phi(s), (\ell, 0))$, where $(\ell, 0) {\in} \mc L(X') {=} \mc L(X) {\times} \mc L(\R^m)$. \end{enumerate}
\end{definition}

\subsection{Putting It Together}
\begin{algorithm}[H]
	\caption{Certificate-finding with regret minimization}
	\begin{algorithmic}[1]\label{alg:rm}
		\STATE {\bf Input:} black-box game, extendable family $\mc A^i$ for each player $i$
		\STATE Initialize confidence sequence $(\hat G^0, \hat \alpha^0, \hat \beta^0)$
		\FOR{$t = 1, 2, \dots$}\label{step:for}
		\STATE Query each $\mc A^i$ to obtain a strategy $\sigma^t_i$
		\STATE Submit loss $-\hat \beta^t_i(\cdot, \sigma^t_{-i})$ to $\mc A^i$
		\STATE Create next pseudogame $\hat G^{t}$ by sampling one playthrough according to $\sigma^t$\label{line:samp}
		\ENDFOR
	\end{algorithmic}
\end{algorithm}

Even in the two-player zero-sum case, this algorithm is {\em not} the exact generalization of \Cref{alg:zerosum}.  That generalization would involve independently solving the lower- and upper-bound games $(\hat G_t, \hat \alpha_t)$ and $(\hat G_t, \hat \beta_t)$ using a total of {\em four} regret minimizers, not two.
This algorithm has no need to store or refer to pessimistic strategies. It suffices to use only the optimistic strategy. As usual when dealing with regret minimization, we will discuss convergence of the average (optimistic) strategy played by each player.
\begin{proposition}\label{prop:correct-rm}
Suppose that the true rewards are bounded in $[0, 1]$. After $t$ iterations of the for loop on Line~\ref{step:for}, assuming the correctness of the confidence  sequence at time $t$,  the average optimistic profile $\bar \sigma^t$ forms a coarse-correlated approximate equilibrium of $G^t$, in which the equilibrium gap for player $i$ is at most
$
\eps_{i,t} = \hat \beta^{*t}_i(\bar \sigma_{-i}^t) - \hat \alpha^t_i(\bar \sigma^t).
$
\end{proposition}
Thus, \Cref{alg:rm} is an anytime algorithm whose equilibrium gap bound at any time $t$ can be easily computed by linear passes through the pseudogame $\hat G^t$.
In the two-player zero-sum case (wherein, for notation, $\beta = \beta_1$ and $\alpha = -\beta_2$ and $\bar\sigma^t = (\bar x^t, \ubar y^t)$), we can use the slightly tighter
$
\eps_t = \hat \beta^{*t}(\ubar y^t) - \hat \alpha^{*t}(\bar x^t)
$
as a Nash gap bound.

\subsection{Convergence Rate}\label{s:convergence-problem}
\begin{figure*}[!t]
\centering
\includegraphics[scale=0.5]{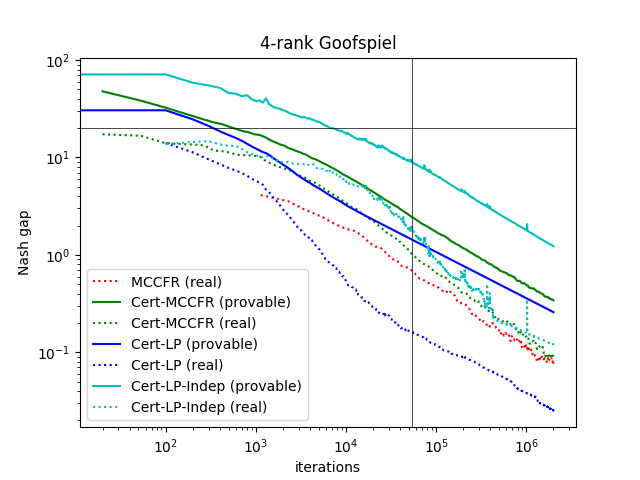}
\includegraphics[scale=0.5]{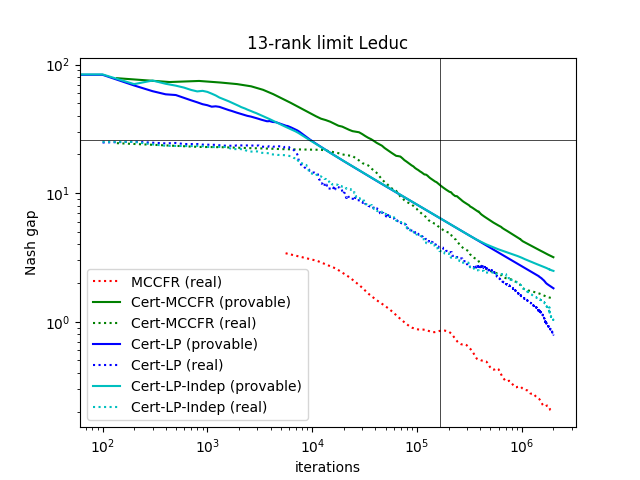}
\includegraphics[scale=0.5]{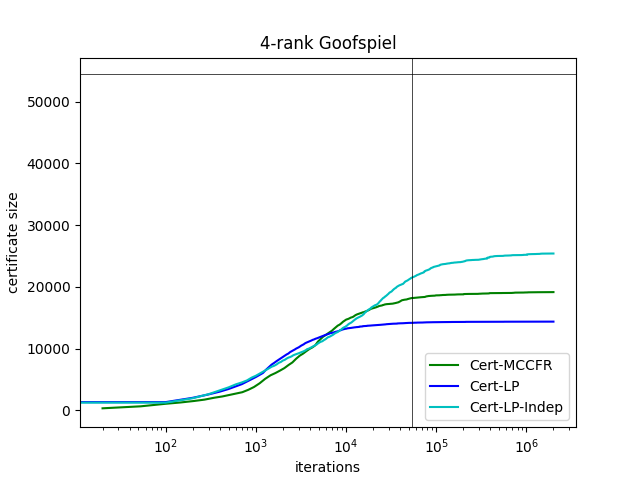}
\includegraphics[scale=0.5]{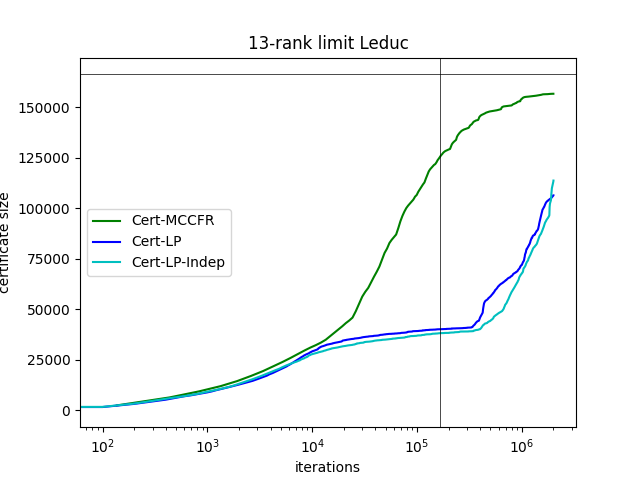}
\caption{Convergence of \Cref{alg:zerosum} and \Cref{alg:rm} in 4-rank Goofspiel and 13-card limit Leduc\protect\footnotemark. To be consistent with the other algorithms, one ``iteration'' of MCCFR consists of one accepted loss vector per player. For the other algorithms, one ``iteration'' is one playthrough. In all cases, we show both the {\em provable} equilibrium gap $\hat \beta^{*t}(\sigma^t) - \hat \alpha^{*t}(\sigma^t)$ and the {\em true} equilibrium gap $\beta^{*t}(\sigma^t) - \alpha^{*t}(\sigma^t)$. The exception is MCCFR, which on its own does not give provable equilibrium gaps in the same way. The horizontal line in the Nash gap graphs is at the game's reward range (Goofspiel has reward range $[-10, 10]$ and 13-rank Leduc has $[-13, 13]$, so the lines are at $20$ and $26$, respectively). The vertical line in both graphs is at the number of nodes in the game (Goofspiel has 54,421 nodes and 13-rank Leduc has 166,366).   }\label{fig:experiments}
\end{figure*}

\begin{table*}
    \centering
    \def\arraystretch{1.1}
    \begin{tabular}{c | c | c}
    & {\bf Sampling-limited} & {\bf Compute-limited} \\\hline
        {\bf Unknown nature distributions} & \Cref{alg:zerosum} with LP solver & \multirow{2}{*}{\makecell{\Cref{alg:rm} with a CFR variant \\ (\textit{e.g.}, outcome-sampling MCCFR)}}
        \\\cline{1-2}
        \makecell{{\bf Known nature distributions}} &  \makecell{Algorithm~6.7 of \citet{Zhang20:Small}}
    \end{tabular}
    \caption{Algorithms we suggest by use case in two-player zero-sum games. {\em Sampling-limited} means that the black-box game simulator is relatively slow or expensive compared to solving the pseudogames. {\em Compute-limited} means that the simulator is fast or cheap compared to solving the pseudogames. In general-sum games, only \Cref{alg:rm} is usable.} 
    \label{tab:summary}
\end{table*}

Annoyingly, it is {\em not} the case in general that $\eps_{i,t} = \tilde O(N_t/\sqrt{t})$.  A counterexample is provided in the appendix. Intuitively, the reason is that, for a fixed strategy $\sigma$, the upper bound $\hat \beta^t(\sigma)$ is not a monotonically nonincreasing function of $t$; indeed, for strategies $\sigma$ that are not sampled very frequently, $\hat \beta^t(\sigma)$ may fluctuate by large amounts even when $t$ is large. However, the nonmonotonicity of $\hat \beta^t$ is, in some sense, necessary to achieve the high-probability correctness guarantee. If $\hat \beta^t$ does not increase over time, then the probability that it is an incorrect bound remains constant, rather than decreasing polynomially with time as would be desired. 

To study the convergence rate of \Cref{alg:rm}, then, we will instead analyze the quantity
\begin{align}
\bar \eps_{i,T} &= \max_{\sigma_i} \frac{1}{T} \sum_{t=1}^T \qty[\hat \beta^t_i(\sigma_i, \sigma^t_{-i}) - \hat \alpha^t_i(\sigma^t)] + O\qty(\frac{1}{\sqrt{T}})
\\&= \frac{1}{T} \qty[R_{i,T} + U_{i,T}] + O\qty(\frac{1}{\sqrt{T}})
\end{align}
where the $O$ hides only an absolute constant.
This quantity is identical to $\eps_{i,t}$ except that it uses $\hat \beta^t_i$ with $\sigma^t_{-i}$ instead of $\hat \beta^T_i$ to match the regret term, and has an extra error term added. 
\begin{proposition}\label{prop:eps-bar-correctness}
With probability $1 - O(1/\sqrt{T})$, $\bar\eps_{i,T}$ is an actual equilibrium gap bound.
\end{proposition}

By \Cref{thm:main}, $U_T = \tilde O(N_T\sqrt{T})$. Thus, this theorem matches the worst-case convergence of any algorithm with regret $\tilde O(N_T/\sqrt{T})$, up to a logarithmic factor. For example, using CFR and variants thereof matches the bound of \Cref{thm:main-lp} with iterates that are linear time in the size of the pseudogame. With MCCFR, the iterates can be made even faster, and due to \citet{Farina20:Stochastic}, even outcome-sampling MCCFR can be used without breaking the $\tilde O(N_T/\sqrt{T})$ runtime bound.

Unfortunately, there is a further problem. It is often unwieldy to compute $\bar \eps_{i,T}$. For example, if using outcome-sampling MCCFR, one may not even have access to the true bounds $\hat \beta^t(\cdot, \sigma^t_{-i})$ (and similar for $\alpha$) but only stochastic estimates $\tilde \beta^t(\cdot, \tilde \sigma^t_{-i})$ with the correct conditional expectation \cite{Farina20:Stochastic}. In that case, the stochastic estimate may be used as a substitute to create a stochastic equilibrium gap bound 
{\footnotesize
\begin{align}
\tilde \eps_{i,T} &= \max_{\sigma_i} \frac{1}{T} \!\sum_{t=1}^T\! \qty[\tilde \beta^t_i(\sigma_i, \sigma^t_{-i}) - \tilde \alpha^t_i(\sigma^t)] 
+O\qty(\!M\sqrt{\frac{1}{T} \log T})
\end{align}
}%
where $M$ is a bound on the norm of the estimates; \textit{i.e.}, $\abs{\tilde \beta^t_i(\sigma_i, \tilde \sigma^t_i) - \tilde \beta^t_i(\sigma'_i, \tilde \sigma^t_i) } \le M$ for every pair of strategies $\sigma_i, \sigma_i'$. As discussed by \citet{Farina20:Stochastic}, with a uniform sampling vector, we can achieve $M \le N_T$. 
\begin{proposition}\label{prop:eps-tilde-correctness}
With probability $1 - 1/T$, for every time $T$ and player $i$, we have $\tilde \eps_{i, T} \ge \bar \eps_{i,T}$.
\end{proposition}

\footnotetext{The MCCFR convergence plots start at $10^3$ iterations because computing the best response requires a relatively expensive pass through the whole game tree, so we only perform it once every $10^3$ iterations}

Thus, in particular, we have:
\begin{corollary}
$
\eps^*_{i,T} := \min(\eps_{i, T}, \tilde \eps_{i, T}) = \tilde O(N_T/\sqrt{T})
$ is an equilibrium gap bound with probability $1 - O(1/\sqrt{T})$.
\end{corollary}
This is the desired result. In practice, $\tilde \eps_{i,T}$ is trivial until $T = \Omega(N_T^2)$, and $\eps_{i, T}$ is almost always a better bound. Thus, in our experiments, we use only $\eps_{i, T}$. For this reason and for clarity, we have not bothered to specify the constants in the big-$O$s. Nevertheless, it is desirable theoretically to be able to define a quantity $\eps^*_{i, T}$ that has both $\tilde O(N_T/\sqrt{T})$ convergence and (high-probability) correctness. As before, the correctness probability can be raised to any inverse-polynomial function of $T$ by a suitable change to \Cref{eq:rho}.

As an equilibrium-finding algorithm, \Cref{alg:rm} is a ``weaker'' version of just running the underlying regret minimizers on the full game: 
instead of each regret minimizer getting access to the true losses, they only get access to an upper bound. However, its main advantage over regret minimization is, as before, its ability to give a equilibrium gap bound that can be computed without full knowledge of the remainder of the game or exact nature action probabilities.

Finally, \Cref{alg:rm} has an unintuitive property.
\begin{warning}\label{rem:warning}
If $\mc A_i$ are stochastic regret minimizers (e.g. MCCFR), instead of submitting $-\hat \beta^t_i(\cdot, \sigma^t_{-i})$, it may be tempting to submit a noisy (sampled) version of $-\beta^t_i(\cdot, \sigma^t_{-i})$. %
Then the {\em actual} equilibrium gap $\beta^{*t}_i(\bar \sigma_{-i}^t) - \alpha^t_i(\bar \sigma^t)$ will converge, but the {\em provable} equilibrium gap $\bar \eps_{i,t}$ may not. A counterexample is provided in the appendix.
\end{warning}

\subsection{The Case of Known Nature Probabilities: MCCFR Without Uniform Sampling}
If the nature probabilities are assumed to be known exactly, \Cref{rem:warning} does not apply, since the actual bounds $(\alpha^t, \beta^t)$ and the sampled bounds $(\hat \alpha^t, \hat \beta^t)$ are the same. Even in this case, \Cref{alg:rm} is still noteworthy: if we run it with outcome-sampling MCCFR, the result is an MCCFR-like algorithm ({\em i.e.,} an equilibrium finder in the black-box case) that operates without an a-priori ``uniform sampling strategy''. Indeed, the iterations only require a uniform sampling strategy over the {\em current pseudogame}, not the full game!

That algorithm is not quite a regret minimizer in the usual sense: its convergence rate depends on the uncertainty of the sampling method, and is tied to the fact that the sampling in Line~\ref{line:samp} of \Cref{alg:rm} uses the current strategy.

\section{Experiments}\label{s:experiments}

We conducted experiments on two common benchmarks:
\begin{enumerate}
    \item {\em $k$-rank Goofspiel}. At each time $t = 1, \dots, k$, both players simultaneously place a {\em bid} for a prize. The prizes have values $1, \dots, k$, and are randomly shuffled. The valid bids are also $1, \dots, k$, each of which must be used exactly once during the game. The higher bid wins the prize; in case of a tie, the prize is split. The winner of each round is made public, but the bids are not. Our experiments use $k = 4$.
    \item {\em $k$-rank heads-up limit Leduc poker} \citep{Southey05:Bayes}, a small two-player variant of poker played with one hole card per player and one community card. Our experiments use a full range of poker ranks $(k = 13)$.
\end{enumerate}
We tested four algorithm variants. Except in the last case, which we will describe, all certificate-finding algorithms assume that the nature distributions are independent of player actions. In Goofspiel, we assume further that the nature distributions are independent of past {\em nature} actions, which is true (nature always plays uniformly at random).
\begin{enumerate}
\item MCCFR with outcome sampling (OS-MCCFR)~\cite{Lanctot09:Monte} ({\tt MCCFR}). This algorithm requires the game tree to be fully expanded, and does not give a (nontrivial) certificate. However, it does give a benchmark for actual equilibrium gap convergence.
\item \Cref{alg:rm} with OS-MCCFR as the regret minimizer ({\tt Cert-MCCFR}).
\item \Cref{alg:zerosum}, with LP for the game solves ({\tt Cert-LP}). Since the LP solves are relatively expensive, we only recompute the LP solution every $100$ playthroughs sampled. This does not change the asymptotic performance of the algorithm. We use Gurobi v9.0.0 \cite{19:Gurobi} as the LP solver.
\item \Cref{alg:zerosum}, except with no assumptions on relationships between nature distributions ({\tt Cert-LP-Indep}).
\end{enumerate}
\Cref{fig:experiments} shows the results. As expected, all the algorithms show a long-term convergence rate of roughly $\tilde \Theta(1/\sqrt{t})$. All certificate-finding algorithms find nontrivial provable certificates with fewer samples than it would take to expand the whole game tree, showing the efficacy of our method.

\section{Conclusion and Future Research}
We developed algorithms that construct high-probability certificates in games with only black-box access. Our method can be used with either an exact game solver (e.g., LP solver) as a subroutine or a regret minimizer such as MCCFR. \Cref{tab:summary} shows which algorithm we recommend based on the use case. As a side effect, we developed an MCCFR-like model-free equilibrium-finding algorithm that converges at rate $O(\sqrt{\log(t)/t})$, and does not require a lower-bounded sampling vector. We are, to our knowledge, the first to obtain this result.
Our experiments show that our algorithms produce nontrivial certificates with very few samples. 

This work opens many avenues for future research.
\begin{enumerate}
    \item Is there a ``cleaner'' way to fix the problem introduced in \Cref{s:convergence-problem}? For example, a different confidence sequence may fix the problem, or it could be the case that $\eps_{i,T}$ is small for {\em most} times $t$ (or even only a constant fraction), which would show that $\min_{t \le T} \eps_{i, T} = \tilde O(1/\sqrt{T})$, matching \Cref{thm:main-lp}.
    \item Is it possible to adapt \Cref{alg:rm} to work with a generic extensive-form iterative game solver, for example, first-order methods such as EGT \cite{Hoda10:Smoothing,Kroer18:Solving}?
    \item In many practical games, there are nature nodes $h$ for which, under a particular profile $\sigma$, every child of $h$ has similar utility: the range of utilities of the children of $h$ under $\sigma$ is far smaller than $[\alpha(h \to *), \beta(h \to *)]$. Is it possible to incorporate this sort of information into the confidence-sequence pseudogames without losing perfect recall (which is needed for efficient solving)?
\end{enumerate}

\newpage
\section*{Acknowledgements}

This material is based on work supported by the National Science Foundation under grants IIS-1718457, IIS-1901403, and CCF-1733556, and the ARO under award W911NF2010081.

\bibliography{dairefs,main}

\clearpage
\onecolumn
\appendix
\section{Proofs of Theorems}
\subsection{\Cref{thm:ucb}}
\begin{lemma}\label{lem:hoeffding}
	Fix a player $i$ and chance node $h$. With probability at least $1 - 2/t^2Cn$, for any assignment $u : \op{Children}(h) \to [\alpha, \beta]$ of utilities, we have \begin{align}
	\abs{\E_{a \sim \sigma_0|h} u(ha) - \E_{a \sim \hat \sigma_0|h} u(ha)} \le  (\beta - \alpha)\rho(h).
	\end{align}
\end{lemma}
\begin{proof}
	If $\rho = 1$ the claim is trivial, so assume $\rho < 1$. The desired error term is a convex function of $u$, so we need only prove the theorem for $u : \op{Children}(h) \to \{ \alpha, \beta\}$. By definition, $\hat \sigma_0|h$ was created by sampling $t(h)$ times. Thus, by Hoeffding, we have
	\begin{align}
	\Pr[\abs{\E_{a \sim \sigma_0|h} u(ha) - \E_{a \sim \hat \sigma_0|h} u(ha)} \ge (\beta - \alpha)\rho]
	&\le 2 \exp(-2t(h)\rho(h)^2)
	\\&= 2 \exp(- \abs{A_h} \log 2  - \log t^2Cn )
	\\&= \frac{2^{1-\abs{A(h)}}}{t^2Cn}
	\end{align}
	Taking a union bound over the $2^{|A_h|}$ choices of $u$ completes the proof.
\end{proof}
Thus, by a union bound, with probability $1 - 2/t^2$, the above lemma is true for every player and chance node. Condition on this event, and take any player $i$ and any profile $\sigma$. For notation, let $\hat \sigma$ be the strategy profile in which chance plays according to $\hat \sigma_0$ and the players play according to $\sigma$.
\begin{lemma}
	At every node $h$, we have the bounds $\hat \alpha_i(\sigma|h) \le \alpha_i(\sigma|h) \le \beta_i(\sigma|h) \le \hat \beta_i(\sigma|h)$.
\end{lemma}
\begin{proof}
	By induction, leaves first. At the leaves, the lemma is trivial. Let $h$ be any internal node. Then we have
	\begin{align}
		\hat \alpha_i(\sigma|h) &= \E_{a \sim \hat \sigma_i|h} \hat \alpha_i(\sigma|ha) - \rho(h) \Delta_i(h \to *)
		\\&\le \E_{a \sim \hat \sigma_i|h} \alpha_i(\sigma|ha) - \rho(h) \Delta_i(h \to *)
		\\&\le \E_{a \sim \sigma_i|h} \alpha_i(\sigma|ha)
		\\&= \alpha_i(\sigma|h).
	\end{align}
where the first two inequalities use, in order, the inductive hypothesis and the last lemma. An identical proof holds for $\beta$, and we are done.
\end{proof}
The theorem now follows by applying the above lemma with $h = \emptyset$.
\subsection{\Cref{thm:main}}
Assume WLOG there is only one player, and drop the subscript $i$ accordingly. Define the {\em sampled cumulative uncertainty} $\hat U_T$ as 
\begin{align}
\hat U_T := \sum_{t=1}^T \hat \Delta^t(z_t)
\end{align}
where $z_t$ is the last node in $\hat G^t$ reached during the play at time $t$. By linearity of expectation, we have $\E \hat U_T = \E U_T$.
Define $\hat U_K^\Delta(h)$ to be the sampled regret at node $h$ after node $h$ is sampled $K$ times. Formally,
\begin{align}
\hat U_K(h) &:= \sum_{k=1}^K \hat \Delta^{t_{h,k}}(h \to z_{t_{h,k}})
\end{align}
where $t_{h,k}$ is the $k$th timestep on which $h$ was sampled. Conveniently, $\hat U_K(h)$ can be analyzed independently of the rest of the game. Our goal is to bound $\hat U_T = \hat U_T(\emptyset)$.

Let $N_k(h)$ be the number of descendants of $h$, including $h$ itself, at time $t_{h,k}$. Let $C_{k}(h)$ be the same, except only counting chance nodes. Let $\rho_k(h)$ be the value of $\rho(h)$ after $k$ samples at $h$. Once again, these quantities are independent of what happens outside the subgame rooted at $h$. We now prove a lemma, which has the theorem as the special case $h = \emptyset$.
\begin{lemma}
For every exploration policy $\mc A$, any node $h$ of $G$, and any time $K$, we have
\begin{align}
\E \hat U_K(h) \le 2 C_{k}(h)\sqrt{2KM} + N_K(h).
\end{align}
\end{lemma}
\begin{proof}
By induction on the nodes of the game tree, leaves first. For each child $ha$ of $h$, let $K_a$ be the number of times action $a$ has been sampled.

{\em Base case.} If $h$ is a leaf of $G$, then uncertainty at most $1$ will be incurred when the leaf is expanded for the first time.

{\em Inductive case.}
\begin{align}
\E \hat U_K(h) &\le \Delta(h \to *) \qty(1 + 2 \sum_{k=1}^K \rho_k(h)) + \sum_{a \in A_h} \E \hat U_{K_a}(ha)
\\&\le 1 + 2 \sum_{k=1}^K  \sqrt{\frac{M}{2k}} + \sum_{a \in A_h} \qty[C_{K_a}(ha)\sqrt{2K_aM} +  N_{K_a}(ha).]
\\&\le 2C_K(h) \sqrt{2KM} + N_K(h)
\end{align}
where the three terms come from:
\begin{enumerate}
\item a regret of at most $1$, incurred when $h$ is first expanded,
\item the regret incurred at $h$ itself, if it is a chance node, and
\item the regret incurred at each child node. \qedhere
\end{enumerate}
\end{proof}

Once again, the theorem is the above lemma applied with $h = \emptyset$.

\subsection{\Cref{prop:cert-lp-eqm} and \Cref{prop:correct-rm}}
Follow immediately from \Cref{thm:ucb}.
\subsection{\Cref{prop:uncertainty-nash-gap}}

Follows immediately from the definition of a pseudogame.
\subsection{\Cref{prop:eps-bar-correctness}}

Taking a union bound over times $t \ge \sqrt{T}$ in \Cref{thm:ucb}, we have that, with probability $1 - O(1/\sqrt{T})$, $\hat \beta^t_i(\sigma_i, \sigma^t_{-i}) - \hat \alpha^t_i(\sigma^t) \ge \beta^t(\sigma_i, \sigma^t_{-i}) - \alpha^t(\sigma^t)$ for all $t \ge \sqrt{T}$. The bound follows.

\subsection{\Cref{prop:eps-tilde-correctness}}

Identical to Theorem~1 of \citet{Farina20:Stochastic}.

\section{Counterexamples}
\subsection{Rate of convergence of the upper bound in \Cref{prop:correct-rm}}\label{s:counterexample-convergence}

Consider the following multi-armed bandit instance with two arms, formulated as a one-player game: the left arm gives loss $-K$ with probability $1/K$, and $0$ with probability $1 - 1/K$. The right arm gives loss $-1$ deterministically.

With probability $\Theta(1/K)$, the first $\Theta(K^2)+1$ samples of the left arm give rewards exactly $(-K, -K, \dots, -K, 0)$. Condition on this event. 

After $\Theta(K^2)$ samples of the left arm, its upper bound will be 
\begin{align}
-K + \Theta\qty(K \sqrt{\frac{1}{K^2} \log T}) = -K + \Theta(\sqrt{\log T})
\end{align}
The $\Theta(K^2)+1$st sample will not happen until the upper bound exceeds at least $-1$, which only happens once $T > \exp(\Theta(K^2))$. Upon taking the $\Theta(K^2)+1$st sample, the upper bound on the left arm's utility will increase by $\Theta(1)$. But the reward range of this game is $[-K, 0]$, so now taking any $K = o(\sqrt{T})$ completes the counterexample.

\subsection{\Cref{rem:warning}}\label{s:counterexample-convergence2}
For example, consider the one-player multi-armed bandit case with two arms of differing utilities $u(L) < u(R)$. Then the following two statements are simultaneously true:
\begin{enumerate}
\item With MCCFR, with probability $1$, there will exist some time $T$ after which $L$ will no longer be played ever again.
\item $\hat \beta^t(L)$ will increase without bound if it is not played.
\end{enumerate}
Thus, eventually, we will have $\hat \beta^t(L) > \hat\beta^t(R)$, after which time the provable equilibrium gap will always be at least their difference.

\end{document}